\documentclass[runningheads]{llncs}
\usepackage[utf8]{inputenc}
\usepackage{amsmath}
\usepackage{amssymb}
\usepackage[linesnumbered,ruled]{algorithm2e}
\usepackage[T1]{fontenc}
\usepackage{graphicx}
\usepackage{todonotes}
\usepackage{mathtools}
\usepackage{comment}
\usepackage{url}
\usepackage{fullpage}
\usepackage{cleveref}

\title{Efficient Algorithms for Sorting in Trees}

\author{Jishnu Roychoudhury \inst{1} \and Jatin Yadav \inst{2}}

\authorrunning{J. Roychoudhury and J. Yadav}

\institute{United World College of South East Asia, Singapore
\and
Amlgo Labs
}

\date{January 2022}

\begin{document}

\maketitle

\begin{abstract}
Sorting is a foundational problem in computer science that is typically employed on sequences or total orders.
More recently, a more general form of sorting on partially ordered sets (or posets), where some pairs of elements are incomparable, has been studied. General poset sorting algorithms have a lower-bound query complexity of $\Omega(wn + n \log n)$, where $w$ is the width of the poset.

We consider the problem of sorting in trees, a particular case of partial orders, and parametrize the complexity with respect to $d$, the maximum degree of an element in the tree, as $d$ is usually much smaller than $w$ in trees. For example, in complete binary trees, $d = \Theta(1), w = \Theta(n)$. We present a randomized algorithm for sorting a tree poset in worst-case expected $O(dn\log n)$ query and time complexity. This improves the previous upper bound of $O(dn \log^2 n)$. Our algorithm is the first to be optimal for bounded-degree trees. We also provide a new lower bound of $\Omega(dn + n \log n)$ for the worst-case query complexity of sorting a tree poset. Finally, we present the first deterministic algorithm for sorting tree posets that has lower total complexity than existing algorithms for sorting general partial orders.
\end{abstract}

\section{Introduction}

Sorting is a foundational problem in computer science that is typically employed on a set or multi-set of elements. Given an input set $S$ of $n$ elements, typical sorting problems determine the underlying total order or linear sequence of the elements. These algorithms assume that the elements are drawn from an ordered domain (such as the domain of integers). Most sorting algorithm use comparisons in which direct comparisons between pairs of elements allows the algorithms to acquire information about the total order~\cite{cormen2022introduction}. Some recent works consider a restricted version of standard comparison sorting where  only a subset of pairs are allowed to be compared~\cite{alon1994matching,banerjee2016sorting,huang2011algorithms,kuszmaul2022stochastic}.

More recently, a more general form of sorting on partially ordered sets (or posets), where some pairs of elements are incomparable, has been studied~\cite{afshar2021parallel,biswas2017improved,cardinal2013generalized,DKMRE07,Faigle}. Given a input set $P$ of $n$ elements, poset sorting algorithms determine the underlying partial order of the elements. Sorting problems are inherently more challenging for partial orders compared to total orders. Sorting algorithms are generally evaluated by two complexity measures: the {\em query complexity}, the number of comparisons involved, and {\em total complexity}, the total number of computational operations involved. As each individual query can potentially be quite expensive (for example, requiring physical experiments), query complexity is equally important as total complexity. Faigle and Turán were the first to consider the problem of sorting partial orders~\cite{Faigle}, providing an algorithm with $O(wn \log n)$ query complexity, where $n$ is the number of elements and $w$ is the width of the poset. Daskalakis et al.~\cite{DKMRE07} improved on this result to present an algorithm with optimal query complexity $O(wn + n \log n)$, and provided a separate algorithm with time complexity $O(w^2 n \log \frac{n}{w})$. 

Our work is the first to consider sorting tree posets, a particular case of partial orders. In a tree, the value of $w$ is typically very high, leading to existing poset sorting algorithms performing poorly. Therefore, it is important to design efficient customized algorithms for sorting in tree posets with lower query and total complexity. We instead consider the parameter $d$, the maximum out-degree of a node in the natural arborescence representation of the tree. In contrast to width $w$, the degree $d$ is typically much smaller; for example, $d = \Theta(1), w = \Theta(n)$ in complete binary trees.

Apart from the theoretical significance of the problem, there are many practical applications, for example, constructing the evolutionary tree of coronavirus strains or general phylogenetic trees, network mapping, among others. In addition, the final sorted tree can be used as input for tree searching algorithms, such as~\cite{Heeringa,treesearch}. Onak and Parys~\cite{treesearch} considered extending the concept of binary search to trees, whereas Heeringa, Iordan, and Lewis~\cite{Heeringa} considered searching in \textit{dynamic} trees. 

The most similar work to ours is by Wang and Honorio ~\cite{wang2019reconstructing}, in which the authors show how to reconstruct a directed tree with $n$ elements and node degree of at most $d$ using path queries. Afshar et al.~\cite{afshar2020reconstructing} studied the parallel query complexity of reconstructing binary trees from simple queries involving their nodes. Although ~\cite{wang2019reconstructing} does not consider the problem of sorting a tree poset, the problem they consider can be formulated as a constrained version of tree poset sorting. A key difference is that in ~\cite{wang2019reconstructing}, a bound on node degree $d$ is given as an input. In our problem formulation, the maximum degree $d$ is not known as an input and instead is discovered by the algorithm itself. This flexibility significantly improves the applicability of the solution. We only use $d$ as a parameter in complexity analysis. Wang and Honorio provide in their paper a randomized algorithm with query complexity $O(dn \log^2 n)$, given $d$ as an input, and prove lower bounds of $\Omega(dn)$ for the query complexity of any deterministic algorithm and $\Omega(n \log n)$ for the query complexity of any randomized algorithm. We provide a randomized algorithm with query and time complexity $O(dn \log n)$ for sorting a tree poset without a priori knowledge of $d$. This result is optimal for trees of bounded degree, such as binary trees. Our randomized algorithm is unrelated to the algorithm presented in ~\cite{wang2019reconstructing} and utilizes a novel technique based on divide and conquer on the centroid of the tree and random sampling of elements. We also improve the lower bounds for randomized and deterministic algorithms to $\Omega(dn + n \log n)$. Therefore, we improve the gap from $O(d \log^2 n)$ to $O(\log n)$. Finally, we present the first deterministic algorithm for sorting tree posets that has a lower time complexity than existing algorithms for sorting general partial orders.

\section{Preliminaries}
\label{sec:prelim}

We present the formal definitions for the concepts related to posets, tree posets, and sorting in posets. We also repeatedly use certain concepts related to graph-theoretic trees, such as children, edges, degree, and root. In these cases, we use the same terms as graph-theoretic trees and formally define the concepts in terms of set theory here for consistency.

\begin{definition}[poset]
A partially ordered set or poset $\mathcal{P} = (P,\succ)$ is a set of elements $P$ together with a binary relation $\succ \; \subset  P \times P$ which is irreflexive, transitive, and antisymmetric.
\end{definition}
For a pair $(a,b) \in \ \succ$, we write $a \succ b$ and state that $a$ dominates $b$, or $b$ is dominated by $a$.

\begin{definition}[tree]
 A tree is a poset $\tau = (T,\succ)$ where for all elements $t \in T$, the set of ancestors of $t$ in the tree, $Anc_t = \{s \in T : s \succ t\}$, is well-ordered by the relation $\succ$.
\end{definition}

\begin{definition}[maximal element]
In a poset $\mathcal{P} = (P,\succ)$, an element $x \in P$ is a maximal element if the set $\{ y \in P : y \succ x \}$ is empty. The set $M_\mathcal{P}$ is the set of maximal elements in $\mathcal{P}$. We may also refer to a maximal element as a \emph{root}.
\end{definition}
Also, we define an element $e \in S$ to be a maximal element of a set $S$ for some $S \subseteq T$ if $e$ is a maximal element of $\mathcal{S} = (S, \succ \cap \: (S \times S))$.

\begin{definition}[parent, child]
In a tree $\tau = (T,\succ)$, a child of an element $x \in T$ is an element $y \in T$ such that $x \succ y$ and the set $\{ z \in T : x \succ z \succ y\}$ is empty. For an element $t \in T$, we define $ch_t$ to be the set of children of $t$. The parent of an element $x \in T \setminus M_\tau$ is an element $y \in T$ such that $y \succ x$ and the set $\{ z \in T : y \succ z \succ x\}$ is empty. We denote this parent as $par_x$. 

\end{definition}
We note that, for any $y \in ch_x$, $par_y = x$, and that every element has a single parent, except maximal elements, which have no parent.


\begin{definition}[edge]
The set of edges for a tree $\tau = (T,\succ)$ is $E_\tau = \{ (a,b) \in \succ : b \in ch_a \}$. Any pair $(a,b) \in E$ is referred to as an edge. There are $|T| - 1$ edges in a tree with a single maximal element.
\end{definition}

\begin{definition}[degree]
In a tree $\tau = (T,\succ)$, we define the degree of an element $t \in T$, $deg_t$, as the number of children of $t$. We define the maximum degree of the tree, $d(\tau) = \max(|M_\tau|, \max_{t \in T} deg_t)$.
\end{definition}
Throughout our paper, we simply use $d$ to denote $d(\tau)$, as for the problem we deal with, the tree $\tau$ is fixed (but unknown).

For simplicity, we only consider trees with one maximal element in the rest of this paper. We can easily modify any tree to have one maximal element. Say the elements of the tree are $\{t_1, t_2, \dots, t_n\}$.  We introduce a new element $t_{n+1}$, such that $t_{n+1} \succ t_i$ for all $i \in \{1, \dots, n\}$, and sort the modified tree. We never need to make a query involving $t_{n+1}$ as it is known that $t_{n+1}$ dominates all other elements.

\begin{definition}[descendants]
The set of descendants of an element $t$ is $Desc_t = \{s \in T : t \succ s\}$. The subtree of an element $t$ is $Sub_t = \{t\} \cup Desc_t$.
\end{definition}

\begin{definition}[poset width]
In a poset $\mathcal{P}$, an antichain $A \subseteq P$ is a subset of mutually incomparable elements. The width $w(\mathcal{P})$ of the poset is the maximum cardinality antichain of the poset $\mathcal{P}$. It should be noted that in a tree, $w(\tau)$ is equal to the cardinality of the set of minimal elements in $T$.
\end{definition}
Similar to the case for maximum degree, throughout this paper we use $w$ to denote $w(\mathcal{P})$ as for the problems we deal with, the poset $\mathcal{P}$ is fixed.

Next, we define the oracle we use for queries before defining the sorting problem on posets and trees.

\begin{definition}[comparison oracle]
A comparison query $\mathcal{Q}$ for a poset $\mathcal{P} = (P,\succ)$ is a function $\mathcal{Q}_{\mathcal{P}} : P \times P \rightarrow \{\succ, \prec, ||\}$ such that

\[
  \mathcal{Q}_{\mathcal{P}}(i,j)  =\left\{
                \begin{array}{ll}
                  \succ \; \text{if} \; \; i \succ j\\
                  \prec \; \text{if} \; \; j \succ i\\
                  || \; \; \text{otherwise}
                \end{array}
              \right.
 \]
The comparison oracle $\mathcal{O}_{\mathcal{P}}$ responds to comparison queries in $O(1)$ time per query.
\end{definition}
For simplicity, we use $\mathcal{Q}$ to denote $\mathcal{Q}_{\mathcal{P}}$ and $\mathcal{O}$ to denote $\mathcal{O}_{\mathcal{P}}$ throughout this paper. We note that since a tree is also a poset, our definitions of $\mathcal{Q}$ and $\mathcal{O}$ apply for trees as well. We use $||$ to denote incomparability.


\begin{definition}
The sorting problem for a poset $\mathcal{P}$ is defined as follows. Given a set of elements $P$, the information that $\mathcal{P} = (P,\succ)$ is a poset, and a comparison oracle $\mathcal{O}_{\mathcal{P}}$, construct a data structure that can correctly respond to comparison queries in $O(1)$ time without using the oracle $\mathcal{O}_{\mathcal{P}}$.
\end{definition}
This definition of the problem was also (implicitly) used in the paper~\cite{DKMRE07}. Essentially, the sorting problem for a poset attempts to reconstruct the binary relation $\succ$ through a minimal number of queries to a comparison oracle. As $|\succ|$ may be large (on the order of $n^2$), we require the sorting algorithm to derive only a compact representation of $\succ$ that, when fully constructed, can be queried  in $O(1)$ time. This is similar to the classical sorting problem, where the algorithm constructs a linear sequence corresponding to the total order rather than finding the relation between every pair of elements. The indices of two elements in the linear sequence can be used to respond to comparison queries in $O(1)$ time. 

We now define the sorting problem for trees analogously.

\begin{definition}
The sorting problem for a tree $\tau$ is defined as follows. Given a set of elements $T$, the information that $\tau = (T,\succ)$ is a tree, and a comparison oracle $\mathcal{O}_{\tau}$,   construct a data structure that can correctly respond to comparison queries in $O(1)$ time without using the oracle $\mathcal{O}_{\tau}$.
\end{definition}
Our strategy for sorting a tree $\tau$ of $n$ elements consists of two stages: first, to find the $n-1$ edges of the tree, and second, to use this information to construct a data structure that can respond to comparison queries in $O(1)$ time. The second stage can be performed using depth-first search. We briefly describe it here and the pseudocode is included in Appendix~\ref{Appendix B}.

We run a depth-first search of the tree from the root in $O(n)$ time. We let $s_v$ denote the start time of the  element $v$ in the depth-first search, and $e_v$ the end time. We note that $u \succ v$ iff $s_u \le s_v \le e_v \le e_u$. Therefore, we can answer comparison queries in $O(1)$ time.

For the remainder of this paper, we focus on finding the edges of $\tau$, as that is sufficient to sort $\tau$.

\section{Upper Bounds}
\label{sec:upper_bounds}
\begin{theorem}
There is a randomized algorithm which sorts a tree of $n$ elements and maximum degree $d$ in $O(d n \log{n})$ expected worst-case query and time complexity.
\end{theorem}
    Note that time complexity supersedes query complexity as each query itself takes constant time. Therefore, we will prove that the time complexity of our algorithm is $O(d n \log n)$. We will begin by describing an algorithm with a time complexity of $O(d n \log n + n \log^2 n)$. Then, we will reduce the complexity to $O(dn \log n)$ by making use of estimations based on random sampling, instead of considering all the elements of the tree.

\section{An $O(d n \log n + n \log^2 n)$ algorithm}
\label{sec:naive}

Our algorithm relies on finding a \textit{good} separating element in the tree and then dividing the tree into disjoint sub-problems around this element, such that no sub-problem has a large size. In this section, we use the centroid of the tree as the separating element.

\begin{definition}[Centroid]
The centroid of a tree $\tau$ with $n$ elements is an element $c$ such that the subtree size of $c$, $|Sub_c| \geq \frac{n}{2}$, but for all $x \in ch_c$, $|Sub_x| < \frac{n}{2}$.
\end{definition}

\begin{lemma}
The centroid exists and is unique.
\end{lemma}
\begin{proof}
Consider the following algorithm. We start at the root, and keep recursing to the child with subtree size $\ge \frac{n}{2}$ until there is no such child. This maintains the invariant that any centroid must be inside the subtree of the current element. This process must terminate as the number of elements is finite, and it can only terminate on a centroid (let it be $c$). Any other centroid must lie in the subtree of $c$, but all children of $c$ have subtree sizes less than $\frac{n}{2}$, so $c$ must be the only centroid.
\end{proof}

\subsection{Algorithm Outline}
\label{5.1}
\noindent Following is a brief outline of our algorithm for sorting a tree of $n$ elements and maximum degree $d$ in $O(dn \log n + n \log^2 n)$ time. The detailed pseudocode of the algorithm, \texttt{GET-EDGES}, appears later.
\begin{enumerate}
    \item Choose a random element $x$ and sort the list of its ancestors.
    \item Binary search on this list to find the minimal element $c$ whose subtree has size at least $\frac{n}{2}$.
    \item If $c$ is the centroid continue to step 4, else go back to step 2.
    \item Divide the elements of the tree (except $c$) into $|ch_c|+1$ subproblems. The subtree of each child of $c$ forms a subproblem and all the remaining elements of the tree except the subtree of $c$ form an additional subproblem. 
    \item Solve the sorting problem recursively for each subproblem.
    \item Add the edges from $c$ to its children, and add an edge from the parent of $c$ (if any) to $c$. 
\end{enumerate}

\subsection{Algorithm Details and Complexity Analysis}

Our algorithm is given the initial set of elements, $T$, and a comparison oracle $\mathcal{O}$ for the tree $\tau = (T, \succ)$. Note that when we are solving a subproblem on a subset $H \subseteq T$ of elements, all of the terms defined in Section~\ref{sec:prelim} should be considered with respect to the tree $\tau^{\prime} = (H, \succ \cap \: (H \times H) )$, and, in addition, we use $n = |H|$.

Let us first define some subroutines, all of which make use of the oracle $\mathcal{O}$. We start with \texttt{GET-ROOT}, an algorithm for finding a maximal element in a set of elements $H \subseteq T$. Suppose we have a set $S$ and a maximal element $x \in S$. When we add element $e$ to $S$, there are two possibilities:
\begin{itemize}
    \item $e \succ x$: $e$ must be a maximal element of $S \cup \{e\}$. 
    \item Otherwise, $x$ is also a maximal element of $S \cup \{e\}$. 
\end{itemize}

\noindent This gives us an algorithm to find a maximal element of any set $H$ in $O(n)$ time. We initialise $S$ to a single element and iteratively add the remaining elements of $H$ to $S$ while maintaining a maximal element of $S$.

\begin{algorithm}[H]
    \SetKwInOut{Input}{Input}
    \SetKwInOut{Output}{Output}
    \SetAlgoLined

    \underline{GET-ROOT}\;
    \Input{A set $H \subseteq T$ }
    \Output{A maximal element of $H$}
    Let $H = \{h_1, h_2, \ldots h_n\}$\\
    $x \leftarrow h_1$\\
    \For{$i = 2, \ldots n$}{
          \If{$h_i \succ x$}{
              $x \leftarrow h_i$
          }
    }
    \Return{x}
\end{algorithm}


We now define the algorithm \texttt{GET-SUBTREES-OF-CHILDREN}, which finds the subtrees of the children of a given element $r$ within the tree. This algorithm is used to test whether a given element is the centroid, as well as to divide the tree into subproblems around the centroid.

\texttt{GET-SUBTREES-OF-CHILDREN} works as follows. We start by initialising $S \leftarrow Desc_r$. We then repeatedly find a maximal element of $S$ by calling \texttt{GET-ROOT}, append its subtree to the output, and remove its subtree from $S$. We repeat until $S = \emptyset$.


\begin{algorithm}[H]
    \SetKwInOut{Input}{Input}
    \SetKwInOut{Output}{Output}
    \SetAlgoLined

    \underline{GET-SUBTREES-OF-CHILDREN}\;
    \Input{A set $H\subseteq T$ having exactly one maximal element, and an element $r \in H$}
    \Output{List of subtrees of the children of $r$}
    $S \gets \{x \in H: r \succ x\}$\\
    $C \gets []$ // An empty list \\
    \While{$S \neq \emptyset$}{
        $x \gets $ \texttt{GET-ROOT}($S$)\\
        $X \gets \{x\} \cup \{y \in S \; | \; x \succ y\}$\\
        Append $X$ to $C$ \\
        $S \gets S \setminus X$
    }
    \Return{C}
\end{algorithm}

\begin{lemma} The algorithm \texttt{GET-SUBTREES-OF-CHILDREN} correctly outputs all the subtrees of the children of a given element $r$ in $O(dn)$ time.
\end{lemma}

\begin{proof}
It is clear that each iteration of lines 5 to 8 takes $O(n)$ time. Moreover, the maximal element we find in each iteration is always a child of $r$, and there are at most $d$ such elements; thus, we make at most $d$ iterations. Therefore, the total time complexity is $O(dn)$.
\end{proof}
\noindent We can test whether an element is a centroid by first finding its subtree in $O(n)$ time and then finding the subtrees of all its children in $O(dn)$ time using \texttt{GET-SUBTREES-OF-CHILDREN}.

We are now ready to describe the algorithm \texttt{GET-CENTROID} for finding the centroid of a tree. \texttt{GET-CENTROID} starts by taking a random element $t$ from the tree and finding the list of ancestors of $t$, $Anc_t$. The list $Anc_t$ is well-ordered by definition, so it may be sorted using a typical $O(n \log n)$ algorithm for sorting a total order. The algorithm then binary searches to find the minimal element $y \in Anc_t \cup  \{t\}$ which satisfies $|Sub_y| \ge \frac{n}{2}$. The algorithm finally checks if $y$ is the centroid: if so, it terminates, and otherwise it repeats from the beginning until it finds a centroid.  

\begin{algorithm}[H]
    \SetKwInOut{Input}{Input}
    \SetKwInOut{Output}{Output}
    \SetAlgoLined

    \underline{GET-CENTROID}\;
    \Input{A set $H\subseteq T$ which has exactly one maximal element}
    \Output{The centroid of the tree $(H, \succ \cap \: (H \times H))$}
    $n \gets |H|$\\
    Choose a random element $t$ from $H$\\
    \While{$t$ is not the centroid}{
        Choose a random element $x$ from $H$\\
        $Y\gets \{x\} \cup \{y \in H \: | \: y \succ x\}$\\
        Sort $Y$ according to the order relation $\succ$ /*$Y$ is well ordered by definition*/ \\
        Let $Y = \{y_1, y_2, \ldots y_k\}$ where $y_1 \succ y_2 \succ \ldots y_{k-1} \succ y_k$
        
        Binary search on the largest $i$, such that the subtree size of $y_i$ is $\geq n/2$. \\
        $t \gets y_i$
    }

    \Return{$t$}
\end{algorithm}

\begin{lemma} The algorithm \texttt{GET-CENTROID} finds the centroid in $O(dn + n \log n)$ expected time.
\end{lemma}

\begin{proof}
    It is easy to see that every iteration of lines $4$ to $10$ takes $O(dn + n \log n)$ time. The sorting of the set $Y$ takes $O(n \log n)$ time and testing whether $t$ is a centroid takes $O(dn)$ time. In addition, in each step of the binary search of line 9 we find the subtree size of some element, which uses $O(n)$ time. Since there are $O(\log n)$ steps of binary search, line 9 uses $O(n \log n)$ time overall. Everything else takes $O(n)$ time in total.
    
    Let the centroid be $c$. With probability at least $\frac{1}{2}, c = x$ or $c \succ x$. This is because the subtree size of $c$ is greater than or equal to $\frac{n}{2}$ by definition, and $x$ is chosen uniformly at random out of the $n$ elements. When this happens, $c \in Y$ and the binary search returns $y_i = c$, as the subtree sizes for the descendants of $c$ are less than $\frac{n}{2}$. Hence, with probability at least $\frac{1}{2}$ we find the centroid. Thus, the expected number of iterations is less than or equal to $2$, and we get the centroid in $O(dn + n \log n)$ expected time.
\end{proof}
\noindent We now present the pseudocode for the overall algorithm \texttt{GET-EDGES}, which was outlined previously in Subsection~\ref{5.1}.

\begin{algorithm}[H]
    \SetKwInOut{Input}{Input}
    \SetKwInOut{Output}{Output}
    \SetAlgoLined

    \underline{GET-EDGES}\;
    \Input{A set $H\subseteq T$ which is either empty or has exactly one maximal element}
    \Output{The set of the $|H|-1$ edges of $(H, \succ \cap \: (H \times H))$}
    $n \gets |H|$ \\
    \If{$n \leq 1$}{
        \Return{$\{\}$}
    }
    $c \gets $ \texttt{GET-CENTROID}($H$) \\
    $\chi \gets$  \texttt{GET-SUBTREES-OF-CHILDREN}($H$, $c$) \\
    $C \gets \{c\} \cup \{x \in H \: | \: c \succ x\}$ \\
    $Y \gets H \setminus C$ \\
    $E \gets $ \texttt{GET-EDGES}($Y$) \\
    \For{$X \in \chi$}{
        $E \gets E \: \cup $ \texttt{GET-EDGES}($X$) \\
        $E \gets E \cup \{(c, \texttt{GET-ROOT}($X$))\}$
    }
    \If{$Y \neq \emptyset$}{
        $p = \min \{y \in Y \: | \: y \succ c\}$ /* $\min$ is defined, as the set of ancestors is well ordered */ \\
        $E \gets E \cup \{(p, c)\}$
    }
    \Return{$E$}
\end{algorithm}

\begin{lemma} The algorithm \texttt{GET-EDGES}(H) correctly outputs the edges of $H$ in worst case expected $O(d n \log n + n \log^2 n)$ time. 
\end{lemma}

\begin{proof}
    Our algorithm first finds the edges inside the subtree of the centroid $c$ by recursively calling \texttt{GET-EDGES} and adding $(c, x)$ to the list of edges for every child $x$ of $c$.
    After this, we add the edges that are outside the subtree. Finally, we add the edge between $c$ and its parent if it exists. Clearly, the algorithm is correct as it returns all the edges. 
    
    Let $Q_d(n)$ be the expected time needed by the algorithm in the worst case (the maximum expected value over all different trees with $n$ elements and maximum degree $\le d$).

    Note that every sub-problem is either a subtree of a child of $c$, which by definition has a size $< \frac{n}{2}$, or is equal to $H \setminus C$, where $C$ is the subtree of $c$. As $|C| \geq \frac{n}{2}$, all sub-problems have sizes $\leq \frac{n}{2}$. Also, note that all sub-problems also have degree $\le d$.
    
    It takes $O(dn + n \log n)$ expected time to get the centroid. Then, we divide into sub-problems by calling \texttt{GET-SUBTREES-OF-CHILDREN} in $O(dn)$. All other steps are done in $O(n)$ total. Thus, it takes $O(dn + n \log n)$ expected time to get the division into sub-problems. Moreover, the function $f(n) = dn + n \log n$ satisfies $f(u) \ge 0, f(u) + f(v) \le f(u + v)$ for all $u, v \ge 1$. Using Theorem~\ref{Randomised DNC}, the overall time complexity is $Q_d(n) = O((dn + n \log n) \log n)$ = $O(dn \log n + n \log^2 n)$.
\end{proof}

\section{Optimizing to $O(dn \log n)$}

We want to decrease the complexity of \texttt{GET-CENTROID} from $O(dn + n \log n)$ to $O(dn)$. In our solution, one part of this optimization is to use random sampling to find estimated subtree sizes in $(\frac{n}{\log n})$, instead of finding exact sizes in $O(n)$. In this process, we might not get the centroid from the binary search because of errors in subtree size computation, but we prove that, with a constant lower bound on probability, we get another form of a \textit{good} separator, which we term a \textit{pseudo-centroid}.

\begin{definition}[Pseudo-Centroid]
    A pseudo-centroid of a tree $\tau = (T, \succ)$ is an element $c$, for which the size of subtree of $c$ is $\geq \frac{n}{4}$, and for any of its children, the size of subtree is $< \frac{n}{2}$. In particular, note that the centroid is also a pseudo-centroid.
\end{definition}

\noindent It is clear that, if we replace the centroid by a pseudo-centroid in the previous algorithm, all the sub-problem sizes are less than or equal to $\frac{3n}{4}$. It turns out that we can find a pseudo-centroid in expected $O(n)$ time using sampling, and this removes $O(n \log n)$ from the complexity of \texttt{GET-CENTROID}. We assume sufficiently large $n$ (greater than some constant) in the discussion to follow to avoid minor details such as the possibilities of $\log n = 0$ or $\frac{1}{\log n} > 1$, and for some inequalities to be true. 
\newline

\noindent Let $p = \frac{1}{\log n}$. For finding a pseudo-centroid efficiently, we make the following changes:
\begin{enumerate}
    \item We do not find the exact subtree size when binary searching. Instead, we first randomly sample a subset $V$ of $H$ by adding each element of $H$ to $V$ with an independent probability of $p$. Let's define $S_i$ to be the size of the subtree of element $i$. Now, we define the sampled subtree size, $F'_i$ to be the number of sampled elements that lie inside the subtree of $i$. Clearly $\mathbb{E}(F'_i) = S_i p$. Also, we define the estimated subtree size of $i$ to be $S'_i = \frac{F'_i}{p} $, so that $\mathbb{E}(S'_i) = S_i$. This reduces the the time taken in subtree size computation by a factor of $\log n$.
    \item Now, we will return the minimal element in $Y$ whose estimated subtree size is greater than or equal to $\frac{n}{2} - \frac{1}{p}$. Here, two things have changed: we are using estimated subtree size instead of subtree size, and we now have $\frac{n}{2} - \frac{1}{p}$ as the threshold as opposed to $\frac{n}{2}$ before. Note that the monotonicity is still maintained, i.e. the estimated subtree size of a child will not be greater than that of its parent. So, we can still perform binary search to do this step.
    \item We do not sort the whole chain of ancestors $Y$, which earlier lead to a sorting overhead of $O(n \log n)$. Instead, we first randomly sample a subset $Z$ of $Y$ by adding each element of $Y$ to $Z$ with an independent probability of $p$, and then sort $Z$. Since $\mathbb{E}(|Z|) = \frac{\mathbb{E}(|Y|)}{\log n}$, the expected sorting overhead is now $O(n)$.
    \item We perform the binary search in two phases:
    \begin{enumerate}
        \item Binary search to find the position of the \texttt{estimated} centroid within $Z$.
        \item Find the position within the expected $O(\log n)$ elements that lie in-between this element and the immediate next(if any) in $Z$
    \end{enumerate}

\end{enumerate}

\noindent Once we have found the pseudo-centroid, the rest of the algorithm remains the same as before.

\begin{algorithm}
    \SetKwInOut{Input}{Input}
    \SetKwInOut{Output}{Output}
    \SetAlgoLined

    \underline{GET-PSEUDO-CENTROID}\;
    \Input{A set $H\subseteq T$ having exactly one maximal element}
    \Output{A pseudo centroid of $(H, \succ \cap \: (H \times H))$}
    $n \gets |H|$ \\
    Choose a random element $t$ from $H$\\
    \While{$t$ is not a pseudo-centroid}{
        $p \gets \ \frac{1}{\log n}$ \\
        $V \gets \{\}$ \\
        
        \For{$i \in H$}{
            Add $i$ to $V$ with probability $p$
        }
        
        Choose a random element $x$ from $H$.\\
        
        $Y\gets \{x\} \cup \{y \in H \: | \: y \succ x\}$\\
        
        $Z \gets \{ \texttt{GET-ROOT}(H)\}$ \\
        
        \For{$i \in Y$}{
            Add $i$ to $Z$ with probability $p$
        }
        
        Sort $Z$ according to the order relation $\succ$ \\
        
        Let $Z = \{z_1, z_2, \ldots z_k\}$ where $z_1 \succ z_2 \succ \ldots z_{k-1} \succ z_k$ \\
        
        Define $S'_j$ to be the estimated size of subtree of $j$, $S'_j = \dfrac{|(\{j\} \cup \{i \in H |  j \succ i\}) \cap V |}{p}$ \\
        
        \If{$S'_{z_1} \geq \frac{n}{2} - \frac{1}{p}$}{
            Binary search to find the largest $i$, such that $S'_{z_i} \geq \frac{n}{2} - \frac{1}{p}$ \\
            
            $L \gets \{z_i\}$\\
            
            \For{$ u \in Y $}{
                \If{($i \neq k$ and $z_i \succ u$ and $u \succ z_{i+1}$) or ($i = k$ and $z_i \succ u$)}{
                    $L \gets L \cup \{u\}$
                }
            }Sort $L$ according to the order relation $\succ$\\
        
            Let $L = \{l_1, l_2, \ldots, l_r\}$, where $l_1 \succ l_2 \ldots l_{r-1} \succ l_r$ \\
            
            Binary search to find the largest $i$, such that $S'_{l_i} \geq \frac{n}{2} - \frac{1}{p}$
            
            $t \gets l_i$
        }
    }
    \Return{$t$}
\end{algorithm}

\begin{lemma} The algorithm \texttt{GET-PSEUDO-CENTROID} finds a pseudo-centroid in $O(dn)$ expected time.

\end{lemma}
\label{pseudo-centroid-thm}

\begin{proof}
   Let the centroid be $c$ (this is the unique centroid, not just a pseudo-centroid). We first prove that, for sufficiently large $n$, if all of the following three events occur in an iteration of our algorithm, we get a pseudo-centroid.
    \begin{enumerate}
        \item $x = c$ or $c \succ x$
        \item $S'_c \geq \frac{n}{2} - \frac{1}{p}$
        \item For all $i \in T$, $S'_i \leq S_i + \frac{n}{10}$
    \end{enumerate}
    
    \noindent Assume the iteration returns $c'$. If $c' = c$, we are done. Else, $c \succ c'$. Now, $S'_{c'} \geq \frac{n}{2} - \frac{1}{p}$ (as the binary search returns $c'$) and $S_{c'} \geq S'_{c'} - \frac{n}{10}$ (event 2).  Adding these two gives $S_{c'} \geq \frac{n}{2} - \frac{1}{p} - \frac{n}{10} = \frac{2n}{5} - \log n \geq \frac{n}{4}$ for sufficiently large $n$. Also, $S_{c'} < \frac{n}{2}$ as $c'$ lies strictly inside the subtree of $c$, and all children of $c$ have sizes $< \frac{n}{2}$. So, for every child $r$ of $c'$, $S_r < S_{c'} < \frac{n}{2}$. Thus, $c'$ is a pseudo-centroid.
    
    We will prove that the probability of these 3 events all happening simultaneously is $\geq \frac{1}{8}$ for sufficiently large $n$. So, let us try to get an upper bound on at least one of these not happening. Let $p_1, p_2, p_3$ be the probabilities of events $1, 2, 3$ \textbf{not} happening respectively. We will establish bounds on all three of them.  $\mathbf{p_1 \leq \frac{1}{2}}$ has already been proved before.

    To get an upper bound on $p_2$, note that  $F'_c$ follows a binomial distribution with parameters $S_c \geq \frac{n}{2}$ and $p = \frac{1}{\log n}$. As any median of a binomial distribution with parameters $n$ and $p$ is $\geq \lfloor np \rfloor$~\cite{kaas1980mean}, we have $S_{c'} = F_{c'} / p \geq \lfloor \frac{pn}{2} \rfloor / p \geq \frac{n}{2} - \frac{1}{p} $ with probability atleast $\frac{1}{2}$. Therefore, $\mathbf{p_2 \leq \frac{1}{2}}$.
    
    Now, we will prove that $p_3 \le \frac{1}{4}$ for sufficiently large $n$. Indeed, consider an element $i$, with $k$ elements in its subtree, $u_1, u_2, \ldots u_k$, where $k = S_i$. Let $X_j$ be an indicator variable for whether $u_j$ was chosen in $V$. All $X_j$'s are independent and identically distributed with $\mathbb{E}(X_j) = p$. Also, $F'_i = X_1 + X_2 + \ldots X_k$, and $\mathbb{E}(F'_i) = kp$. Since we want a total error of $\le \frac{np}{10}$ in $F'_i$, we use the Chernoff bound~\cite{chernoff1952measure} with $\mu = kp, \delta = \frac{n}{10 k} \ge \frac{1}{10}$  to get:

    \begin{align*}
    & \displaystyle \mathbb{P} \left(F'_i > kp + \dfrac{np}{10} \right) = \mathbb{P} \left(F'_i > \mu(1 + \delta) \right) \leq \exp \left( -\dfrac{\mu}{3} \delta \min(1, \delta) \right) \leq \exp \left( -\dfrac{kp}{3} \dfrac{n}{10 k} \dfrac{1}{10} \right) = \exp \left(\dfrac{-np}{300} \right)&&\\\nonumber
    \end{align*}

    \noindent But $F'_i = p S'_i$, and $k = S_i$. Therefore, for any $i$,

    \begin{align*}
    & \displaystyle \mathbb{P}\left(S'_i > S_i + \dfrac{n}{10} \right) = \mathbb{P} \left(F'_i > kp + \dfrac{np}{10} \right) \leq \exp \left(\dfrac{-np}{300} \right)&&\\\nonumber
    \end{align*}

    \noindent Using the union bound, $p_3 \leq \displaystyle n \cdot \exp \left(\dfrac{-np}{300} \right) =  n \exp \left(\dfrac{-n}{300 \log n} \right) \leq \frac{1}{4}$ for sufficiently large $n$. 
    
    \vspace{3mm}
    
    Now, let us say an iteration fails if at least one of the events $1, 2, 3$ is false. The probability of failure of an iteration is the sum of the probability of event $1$ not happening and the probability of event 1 happening but at least one of the events 2 and 3 not happening. Noticing that event $1$ is independent of both events $2$ and $3$, and applying union bound,we get that the probability of failure is at max:
    
    
    $$ p_1 + (1-p_1)(p_2 + p_3) \leq  p_1 + \frac{3}{4}(1 - p_1) = \frac{3}{4} + \frac{p_1}{4} \leq \frac{7}{8}$$

\noindent Therefore, we find a pseudo-centroid with probability $\geq \frac{1}{8}$ in every iteration, and hence the expected number of iterations is $O(1)$.

\noindent We now prove that each iteration takes $O(dn)$ time. We divide the iteration into its constituent steps to do so.

\begin{enumerate}
    \item Finding $Y$ clearly takes $O(n)$ time.
    \item Sorting $Z$ takes $O(|Z| \log |Z|$) time, and since $\mathbb{E}(|Z|) \leq \frac{n}{\log n}$, this also takes $O(n)$ time.
    
    \item Finding $S'_j$ takes expected $\mathbb{E}(|V|) = \frac{n}{\log n}$ time. We find $S'_j$ $O(\log n)$ times during the binary search, so the binary search also takes $O(n)$.
    \item Finding the set $L$ clearly takes $|Y| = O(n)$ time, as for each element in $Y$, we query whether it is between two given elements of $Z$.
    \item The expected number of elements in $L$ is $O(\log n)$. This is because the the output $t$ of an iteration is independent of $Z$ and the size of $L$ is equal to the number of elements in $Y$ that lie between the element just below and just above $t$ in $Z$. Hence, the expected size is clearly $O(\log n)$. As a result, expected time to sort $L$ is $\mathbb{E}(|L| \log |L|)$ = $O(\mathbb{E}(|L| \log n))$ = $O(\log n \log n) = O(n)$.
    
    \item We then binary search on the list $L$. In each step of the binary search, we find the estimated subtree size of some element, which takes $O(\frac{n}{\log n})$ time. There are $O(\log n)$ steps, so this step is also $O(n)$.
    
    \item Checking whether $t$ is a pseudo-centroid takes $O(dn)$ time, as described previously.
\end{enumerate}
    
\end{proof}

\begin{lemma} $\texttt{GET-EDGES-OPTIMIZED}$ sorts a tree using $O(dn \log n)$ expected time
\end{lemma}

\begin{proof} The only difference between \texttt{GET-EDGES} and \texttt{GET-EDGES-OPTIMIZED} is that one uses \texttt{GET-CENTROID} and the other uses \texttt{GET-PSEUDO-CENTROID}. 
    
    Here, each sub-problem has size $\leq \frac{3n}{4}$, and we get the division into sub-problems using $O(d n)$ expected time. The function $f(n) = dn$ satisfies $f(u) \geq 0, f(u) + f(v) \le f(u + v)$ for all $u, v \ge 1$. Hence, using Theorem~\ref{Randomised DNC}, the expected time complexity is $Q_d(n) = O(dn \log n)$.
    
\end{proof}

\begin{algorithm}
    \SetKwInOut{Input}{Input}
    \SetKwInOut{Output}{Output}
    \SetAlgoLined

    \underline{GET-EDGES-OPTIMIZED}\;
    \Input{A set $H\subseteq T$ which is either empty or has exactly one maximal element}
    \Output{The set of the $|H|-1$ edges of $(H, \succ \cap \: (H \times H))$}
    $n \gets |H|$ \\
    \If{$n \leq 1$}{
        \Return{$\{\}$}
    }
    $c \gets $ \texttt{GET-PSEUDO-CENTROID}($H$) \\
    $\chi \gets$  \texttt{GET-SUBTREES-OF-CHILDREN}($H$, $c$) \\
    $C \gets \{c\} \cup \{x \in H \: | \: c \succ x\}$ \\
    $Y \gets H \setminus C$ \\
    $E \gets $ \texttt{GET-EDGES-OPTIMIZED}($Y$) \\
    \For{$X \in \chi$}{
        $E \gets E \: \cup $ \texttt{GET-EDGES-OPTIMIZED}($X$) \\
        $E = E \cup \{(c, \texttt{GET-ROOT}($X$))\}$
    }
    \If{$Y$ is not empty}{
        $p = \min \{y \in Y \; | \; y \succ c\}$ /* min is defined, as the set of ancestors is well ordered */ \\
        $E \gets E \cup \{(p, c)\}$
    }
    \Return{$E$}
\end{algorithm}

\section{Deterministic Algorithm}

\begin{theorem}
There is a deterministic algorithm which sorts a tree of $n$ elements and width $w$ in $O(wn + n \log n)$ worst-case query and time complexity.
\end{theorem}

\noindent We briefly describe a deterministic algorithm \texttt{GET-EDGES-DET} which runs in worst-case total complexity of $O(wn + n \log n)$. The pseudocode of the algorithm is included in Appendix~\ref{Appendix B}.





We can find a single minimal element in a set of elements $H \subseteq T$ in $O(n)$ time by using a procedure similar to \texttt{GET-ROOT} presented earlier. Let $H = \{h_1, h_2, \dots, h_n \}$. We initialize $x \gets h_1$. Then, we iterate through all other $h_i$, setting $x \gets h_i$, if $x \succ h_i$. At the end of this procedure, $x$ is a minimal element.

Our algorithm starts with $H = T$. After we find a minimal element $x \in  H$, we can find the set of ancestors of $x$, $Anc_x$, in $O(n)$ time. As $Anc_x$ is a total order, we can sort it in $O(n \log n)$ time. We may then find the parent of each element in the chain $R = Anc_x \cup \{x\}$ except the maximal element, and add the chain to a chain decomposition $\mathcal{C}$. We set $H \gets H \setminus R$, and repeat until $H$ is empty.

The preceding method successfully divides the tree $\tau = (T, \succ)$ into $|\mathcal{C}|$ disjoint chains such that $par_x$ is defined for every element $x$ on a chain except the maximal element. We therefore notice that there are only $|\mathcal{C}|$ elements for which $par_x$ is not defined. Hence, for each of these elements $e$, we use $O(n)$ time to find the minimal element which dominates $e$ (if one exists), which is the parent of $e$.

Finally, for all elements $t \in T$, we construct the edge $(par_t,t)$ if $t$ is not a maximal element.

\begin{lemma}
The algorithm \texttt{GET-EDGES-DET}($T$\!) correctly outputs the edges of $T$ in $O(wn + n \log n)$ time.
\end{lemma}

\begin{proof}
Finding a single minimal element clearly takes $O(n)$ time.  Additionally, the width of $\tau$ is $w$. Therefore, there are exactly $w$ minimal elements in $T$. We observe that each time we find a minimal element of $H$, said element is also a minimal element of $T$. This is because, when we remove a chain $C$ from the tree, no element in $H \setminus C$ succeeds an element in $C$. Therefore, any minimal element in $H \setminus C$ was also a minimal element in $H$. Thus, we make $w$ iterations, which takes $O(wn)$ time. Moreover, $|\mathcal{C}| = w$ as a result.

Finding the minimal element which dominates an element $e$ takes $O(n)$ time. We repeat this $|\mathcal{C}|$ times, once for each chain in the chain decomposition. Thus, this step also takes $O(wn)$ time. Finally, the total complexity of sorting all chains is $O(n \log n)$ as the chains are disjoint. This results in the final time complexity of $O(wn + n \log n)$.

We comment that during our iteration, it is possible that removing elements causes the tree to have multiple roots. However, note that this does not affect the correctness or complexity of our algorithm.
\end{proof}

\section{Lower Bounds}
\label{sec:low}

\begin{theorem}
\label{lb1}
The worst-case expected query complexity of the tree sorting problem on a tree of $n$ elements with maximum degree $d$ is $\Omega(dn + n \log n)$.
\end{theorem}

\begin{proof}
We consider a tree $\tau = (T,\succ)$ on $n$ elements with the following form. Let $T = \{t_1, \dots, t_n\}$. There is a single maximal element $t_1$. In addition, each element $t_i$ ($i>1$) has parent $t_{d \lfloor (i-2)/d \rfloor + 1}$.

We assume that the algorithm has not explicitly queried a pair of elements $(x,y)$ where $par_x = par_y = p$ and $|ch_x| = |ch_y| = 0$. We consider the tree $\tau^{\prime} = (T,\succ^{\prime})$ where $\succ^{\prime} = (\succ \cup \: \{(x, y)\})$. In graph-theoretic terms, $\tau^{\prime}$ is identical to $\tau$ except that the parent of $y$ is $x$ instead of $p$. We note that the answers to queries to $\tau^{\prime}$ are the same as the answers to queries to $\tau$ except for the query $(x,y)$. Therefore, the algorithm is unable to distinguish $\tau$ from $\tau^{\prime}$ without explicitly querying $(x,y)$. There are at least $\lfloor \frac{n-1}{d} \rfloor \binom{d - 1}{2} = \Omega(dn)$ unordered pairs $(x,y)$ which satisfy the condition, thus implying a lower bound on the query complexity of sorting a tree of $\Omega(dn)$.

Sorting a total order has an average case query complexity of $\Omega(n \log n)$~\cite{cormen2022introduction}. Since a total order is a tree with maximum degree $1$, this also implies a lower bound of $\Omega(n \log n)$ for the query complexity of sorting a tree. Combining this result with the previous result gives the lower bound of $\Omega(dn + n \log n)$.

We note that this proof holds for any Las Vegas randomized algorithm as well as deterministic algorithms, since we provide an explicit minimal set of queries which the algorithm is required to make. If any query in this set is not made, the algorithm cannot determine the correct answer with probability $1$.
\end{proof}

\noindent We now consider a more fundamental problem of recognizing whether a given poset is a tree poset by asking comparision queries. We prove the lower bound of $\Omega(n^2)$ for this problem.

\begin{theorem}
The expected worst-case query complexity of recoginizing whether a given poset $\tau = (T, \succ)$ is a tree is $\Omega(n^2)$
\end{theorem}
\begin{proof}
We consider the same tree $\tau$ as used for the proof of Theorem~\ref{lb1}, with $d = 2$. Let $L$ be the set of minimal elements of $\tau$. Note that $|L| = \Theta(n)$. We assume that the algorithm has not explicitly queried a pair of elements $(x,y)$ where $x \in L, y \in L, |Anc_x| < |Anc_y|$. Note that this implies $Anc_x \subset Anc_y$ according to the structure of our tree.  We consider $\tau^{\prime} = (T,\succ^{\prime})$ where $\succ^{\prime} = (\succ \cup \: \{(x, y)\})$ and prove that $\tau^{\prime}$ is a poset. 

Note that all properties of a poset are immediately satisfied except transitivity. Assume that $\succ^{\prime}$ is not transitive. Say $a \succ^{\prime} b \succ^{\prime} c$, but $(a, c)  \notin \: \succ^{\prime}$. Note that either $(a, b) = (x, y)$, or $(b, c) = (x, y)$, else $(a, b)$ and $(b, c)$ were also in $\succ$, implying $(a, c) \in \: \succ$, and hence $(a, c) \in \: \succ^{\prime}$. Also, $(a, b) \neq (x, y)$, as $y$ is a minimal element, contradicting $b \succ^{\prime} c$. Thus $(b, c) = (x, y)$, but then $a \in Anc_x \subset Anc_y$, meaning $a \succ c$, and hence $a \succ^{\prime} c$, which contradicts our assumption.

Let $par_y = p$ in $\tau$. Clearly, $\tau^{\prime}$ is not a tree because $p \succ^{\prime} y$ and $x \succ^{\prime} y$, but $p$ and $x$ are incomparable. We note that the answers of queries to $\tau^{\prime}$ are the same as the answers to  $\tau$ except for the query $(x,y)$. Therefore, the algorithm is unable to distinguish $\tau$ from $\tau^{\prime}$. Since $\tau$ is a tree and $\tau^{\prime}$ is not a tree, the algorithm is therefore unable to determine whether the poset described by the oracle is a tree or not without explicitly querying $(x,y)$. There are at least $\binom{|L|}{2} - 1 = \Omega(n^2)$ queries of this form.

Similar to Theorem~\ref{lb1}, this proof holds for any Las Vegas randomized algorithm.
\end{proof}

\noindent Informally, this result shows that the initial information that $\tau$ is a tree (and not merely a general poset) is necessary for an algorithm to sort $\tau$ efficiently.

\section{Concluding Remarks}

In this paper, we have considered the problem of sorting in trees, a particular case of partial orders. Existing algorithms for sorting in posets perform poorly in trees because their complexity bounds are dependent on the width of the partial order, which is expected to be quite large for trees. We presented a  randomized algorithm for sorting a tree poset in worst-case expected $O(dn\log n)$ query and time complexity, where $d$ is the maximum degree of the tree. This is optimal for trees of bounded degree. We also provided a new lower bound of $\Omega(dn + n \log n)$ for the worst-case query complexity of any randomized or deterministic algorithm for sorting a tree poset. We presented the first deterministic algorithm for sorting tree posets with a lower time complexity of $O(wn + n \log n)$ (where $w$ is the width of the tree) than existing algorithms for sorting general posets. Finally, we showed that the recognition of a tree poset takes at least quadratic time.

\subsubsection{Acknowledgements} We would like to thank Diptarka Chakraborty (National University of Singapore) for providing advice on this work.


\bibliographystyle{plain}
\bibliography{ref}

\appendix

\section{Appendix: Theorems}

\noindent We present the proof of Theorem~\ref{Randomised DNC}, which is used in the complexity analysis of the algorithms \texttt{GET-EDGES} and \texttt{GET-EDGES-OPTIMIZED}.

\begin{theorem}\label{Randomised DNC} Consider a randomized divide and conquer algorithm, that divides an instance $I$ of size $n$ into sub-problems of size $\leq c n$ for some constant $0 < c < 1$ in expected $O(f(n))$ time, such that the sum of sizes of sub-problems is $\le n$, for some $f:\mathbb{N} \rightarrow \mathbb Z_{\ge 0}$, which satisfies $f(u) + f(v) \leq f(u + v)$ for all $u, v \ge 1$. The expected worst case time complexity of the algorithm is $O(f(n) \log n)$
\end{theorem}

\begin{proof}
Consider the recursion tree. Every time we divide a problem, the sub-problem has size $\le c$ times the parent problem. Thus, there are $O( \log n)$ levels in the recursion tree. Consider one level. Let the sizes of the problems at this level be $n_1, n_2, \ldots n_k$. It is easy to see using induction that the sum of sizes at any level is $\le n$. So, let $\sum n_i = r \le n$. The contribution of this level is then $\sum f(n_i) \le \sum f(n_i) + f(1) \cdot (n - r) \le f(n)$ (using $f(u) + f(v) \le f(u + v)$ for all $u, v \ge 1$). Therefore, the total complexity is $O(f(n) \log n)$.
\end{proof}

\newpage

\section{Appendix: Algorithms}\label{Appendix B}

We present the pseudocode of algorithms that were not included in the main body of the paper. We first introduce \texttt{CONSTRUCT-REPRESENTATION}, which takes as input the edges of $\tau=(T,\succ)$ and returns a data structure of $O(n)$ space which succinctly represents the binary relation $\succ$. \texttt{CONSTRUCT-REPRESENTATION} makes use of the subroutine \texttt{DFS}.

\begin{algorithm}[H]
    \SetKwInOut{Input}{Input}
    \SetKwInOut{Output}{Output}
    \SetAlgoLined
    
    \underline{CONSTRUCT-REPRESENTATION}\;
    \Input{A set $T$, an adjacency list $A$ with the outgoing edges from each element in $T$}
    \Output{A pair of arrays $(s,e)$}
    $t \gets 1$ \\
    Let $s$ and $e$ be integer arrays of $|T|$ elements \\
    Let $V$ be a boolean array of $|T|$ elements \\
    Set all elements in $V$ to $0$ \\
    $r \gets$ \texttt{GET-ROOT}($T$) \\
    \texttt{DFS}($r$) \\
    \Return{$(s,e)$}
\end{algorithm}

\begin{algorithm}[H]
    \SetKwInOut{Input}{Input}
    \SetKwInOut{Global}{Global}
    \SetAlgoLined
    
    \underline{DFS}\;
    \Global{A set $T$; an adjacency list $A$ with the outgoing edges from each element in $T$; an integer $t$, initialised to $1$; two integer arrays $s$ and $e$ of $|T|$ elements; a boolean array $V$ of $|T|$ elements, all initialised to $0$}
    \Input{A single element $r$}
    $V_r \gets 1$ \\
    $s_r \gets t$ \\
    $t \gets t+1$ \\
    Let $A_r = \{ a_1, \dots, a_k\}$ \\
    \For{$i=1, \dots, k$}{
        \If{$V_{a_i} = 0$}{
            \texttt{DFS}($a_i$)
        }
    }
    $e_r \gets t$ \\
    $t \gets t+1$ \\
\end{algorithm}

\noindent We also present \texttt{COMPARISON-QUERY}, which uses the data structure produced by \texttt{CONSTRUCT-REPRESENTATION} to answer comparison queries in $O(1)$ time per query.

\begin{algorithm}[H]
    \SetKwInOut{Input}{Input}
    \SetKwInOut{Output}{Output}
    \SetKwInOut{Global}{Global}
    \SetAlgoLined

    \underline{COMPARISON-QUERY}\;
    \Global{Two integer arrays $s$ and $e$}
    \Input{Two elements $u$ and $v$}
    \Output{$\mathcal{Q}(u,v)$}
    \uIf{$s_u \leq s_v$ and $e_v \leq e_u$}{
        \Return{$\succ$}
    }
    \uElseIf{$s_v \leq s_u$ and $e_u \leq e_v$}{
        \Return{$\prec$}
    }
    \Else{
        \Return{$||$}
    }
\end{algorithm}

\newpage \noindent Finally, we present \texttt{GET-EDGES-DET}, which sorts a tree of $n$ elements and width $w$ in $O(wn + n \log n)$ time.

\begin{algorithm}[H]
    \SetKwInOut{Input}{Input}
    \SetKwInOut{Output}{Output}
    \SetAlgoLined

    \underline{GET-EDGES-DET}\;
    \Input{A set $T$}
    \Output{The set of the $|T|-1$ edges of $(T,\succ)$}
    Let $T = \{t_1, \dots, t_n\}$ \\
    $H \gets T$ \\
    $\mathcal{C} \gets \{ \}$ // The chain decomposition \\
    \While{$H \neq \emptyset$}{
        Let $H = \{h_1,h_2,\dots,h_m\}$ \\
        $x \gets h_1$ \\
        \For{$i = 2, \ldots, m$}{
              \If{$x \succ h_i$}{
                  $x \gets h_i$
              }
        }
        $c \gets []$ //An empty list \\
        \For{$i=1, \ldots, m$}{
            \If{$x = h_i$ or $h_i \succ x$}{
                Append $h_i$ to $c$ \\
                $H \gets H \setminus \{h_i\}$
            }
        }
        Sort $c$ according to the relation $\succ$ \\
        $\mathcal{C} \gets \mathcal{C} \cup \{c\}$  \\
    }
    Let $\mathcal{C} = \{c_1,\dots,c_w$\} \\
    $E \gets \{\}$ //The set of edges \\
    \For{$i=1, \ldots, w$}{
        Let $c_i = \{ d_1, \dots, d_{|c_i|}\}$ \\
        Let $t_{n+1} \succ t_i$ for all $i \in \{1,\dots,n\}$ \\
        $x \gets t_{n+1}$ \\
        \For{$i=1, \ldots, n$}{
            \If{$t_i \succ d_1$ and $x \succ t_i$}{
                $x \gets t_i$
            }
        }
        \If{$x \neq t_{n+1}$}{
            $E \gets E \cup \{(x,d_1)\}$
        }
        
        \For{$i=1, \dots, |c_i| - 1$}{
            $E \gets E \cup \{(d_i, d_{i+1})\}$
        }
    }
    \Return{E}
\end{algorithm}

\end{document}